\documentclass[letterpaper, DIV=11]{scrartcl}

\usepackage{amssymb, amsthm, mathtools,bbm}
\usepackage[square,sort,comma,numbers]{natbib}
\newtheorem{theorem}{Theorem}
\newtheorem{corollary}[theorem]{Corollary}
\newtheorem{lemma}[theorem]{Lemma}
\newtheorem{definition}[theorem]{Definition}

\numberwithin{equation}{section}
\numberwithin{theorem}{section}

\newcommand{\tr}{{\operatorname{Tr}\,}}

\newcommand{\beq}[1]{\begin{equation} \label{#1}}
\newcommand{\eeq}{\end{equation}}

\newcommand{\arcosh}{\operatorname{arcosh}}

\begin{document}
\addtokomafont{author}{\raggedright}
\title{\raggedright Phase Diagram of the\\ Quantum Random Energy Model}
\author{\hspace{-.075in} Chokri Manai and Simone Warzel}
\date{\vspace{-.3in}}
\maketitle

\minisec{Abstract} We prove Goldschmidt's formula [Phys.\ Rev.\ B 47 (1990) 4858]  for the free energy of the quantum random energy model.
In particular, we verify the location of the first order and the freezing transition in the phase diagram. The proof is based on a combination of variational methods on the one hand, and 
percolation bounds on large-deviation configurations in combination with simple spectral bounds on  the hypercube's adjacency matrix on the other hand.
\bigskip
\bigskip

\section{Introduction} 

The quantum random energy model  (QREM) draws its motivation from various directions. 
In mathematical biology, it has been put forward as a simple model for mutation of genotypes in a random fitness landscape \cite{ES77,BW01}.
More recently, it gained attention as a basic testing ground of quantum annealing algorithms for searches in unstructured energy landscapes (cf.~\cite{JKrKuMa08,BFKSZ13} and references therein) as well as in the context of many-body localization \cite{LPS14,BLPS16,Bur17,FFI19,SK+19}.
Its original motivation stems from the quest of understanding quantum effects in mean-field spin glasses \cite{Gold90,TD90,ONS07,Craw07,SIC12}. 

The classical backbone, the random energy model (REM) was put forward by Derrida~\cite{Der80,Der81} in the early 1980s as the limiting and solvable case of a class of mean-field spin glasses.
The space of $ N $-bit strings $  \mathcal{Q}_N =  \{ -1, 1\}^N $  serves as the configuration space of the REM. The energy associated with 
$ \sigma = (\sigma_1 , \dots , \sigma_N) \in  \mathcal{Q}_N $ is a rescaled Gaussian random variable 
$$
U(\sigma) := \sqrt{N} \, g(\sigma) 
$$
with $ g(\sigma) $ forming an independent and identically distributed (i.i.d.) process with standard normal law.  
$ \mathcal{Q}_N $ may be interpreted as the state space of a system of $ N $ spin-$\tfrac{1}{2}$ quantum objects recorded, e.g., in the $ z $-basis. 
The corresponding Hilbert space is given by the $N$fold tensor product $ \otimes_{j=1}^N \mathbbm{C}^2$ which is unitarily equivalent to $ \ell^2( \mathcal{Q}_N) $. 
Effects of a transversal (e.g. in the negative $ x $-ditrection) constant magnetic field of strength $ \Gamma \geq 0 $ on the spins are  taken in to account through the componentwise 
flip operators $ F_j \sigma := (\sigma_1, \dots , - \sigma_j  , \dots , \sigma_N ) $, which are implemented on $ \psi \in  \ell^2( \mathcal{Q}_N)  $ as
$$ \left( T\psi\right)(\sigma) := -  \sum_{j=1}^N \psi( F_j \sigma ) .
$$
This operator coincides with the negative sum of $ x $-components of the Pauli matrices. 
The energy of the QREM is then given by an Anderson-type random matrix
\begin{equation}\label{eq:Ham}
	 H := \Gamma \ T + U 
\end{equation}
where $ U $ acts as a multiplication operator  on $ \ell^2( \mathcal{Q}_N)  $.

The process $ U(\sigma) $  is the limiting case $ p \to \infty $ of the Gaussian family of $ p $-spin models characterized by its mean 
and covariance function, 
\begin{equation}\label{eq:spinp}
\mathbb{E}\left[U(\sigma) \right] = 0 , \qquad  \mathbb{E}\left[U(\sigma) U(\sigma') \right] = N \left( N^{-1} \sum_{j=1}^N \sigma_j \sigma_j' \right)^p =: N \xi_p(\sigma,\sigma')  . 
\end{equation}
The case $ p= 2 $ corresponds to
the famous Sherrington-Kirkpatrick model. The simplifying feature of the limit $ p \to \infty $ is the lack of correlations.
The quantum $ p $-spin generalisation of the QREM is then given by the random matrix~\eqref{eq:Ham} in which $ U $ is a multiplication operator by the correlated field.\\

\subsection{Main result}
In this paper, we will be interested in thermodynamic properties of the QREM which are encoded in 
its partition fuction 
$$ Z(\beta, \Gamma) := 2^{-N}\,  \tr e^{-\beta H} $$
at inverse temperature $ \beta \in [0, \infty) $, or, equivalently,  its pressure 
\begin{equation}
p_N(\beta, \Gamma) :=  N^{-1} \; \ln Z(\beta, \Gamma) .
\end{equation}
Up to a factor of $ - \beta^{-1} $, the latter coincides with the specific free energy.

In the thermodynamic limit $ N \to \infty $ the pressure of the REM  converges almost surely \cite{Der80,Der81,Bov06},
\begin{equation}\label{eq:REMc}
\lim_{N\to \infty} p_N(\beta, 0) = p^{\mathrm{REM}}(\beta) = \left\{ \begin{array}{l@{\quad}r} \tfrac{1}{2} \beta^2 & \mbox{if } \; \beta \leq \beta_c , \\[1ex]   \tfrac{1}{2} \beta_c^2 + (\beta - \beta_c) \beta_c & \mbox{if } \;   \beta > \beta_c .\end{array} \right.
\end{equation}
It exhibits a freezing transition into a low-temperature phase characterized by the vanishing of the entropy above 
$$ \beta_c := \sqrt{2 \ln 2 } . $$
\begin{figure}[ht]
\begin{center}
\includegraphics[width=\textwidth] {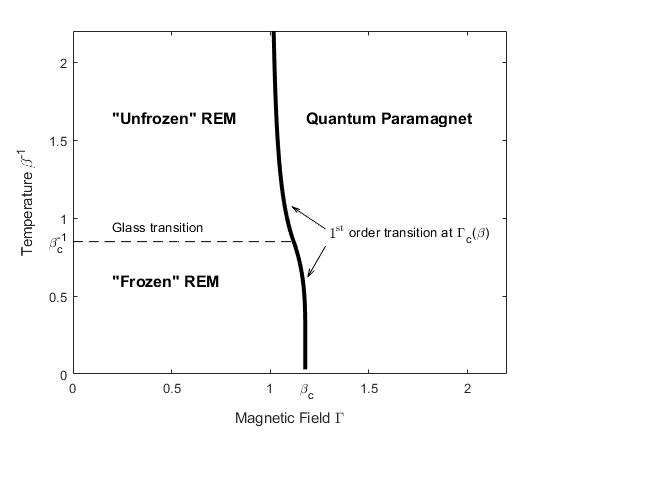}
\end{center}
\vspace*{-1cm}
\caption{Phase diagram of the QREM as a function of the transversal magnetic field $ \Gamma  $ and the temperature $ \beta^{-1}$. The first-order transition occurs at fixed $ \beta $ and $ \Gamma_c(\beta) $. The freezing transition is found at temperature $ \beta_c^{-1} $, which is unchanged in the presence of small magnetic field. } \label{fig:phase}
\end{figure}

Under the influence of the transversal field,  the spin-glass phase of the REM disappears for large $ \Gamma > 0 $ and a first-order phase transition into a quantum paramagnetic phase characterised by
$$ p^{\mathrm{PAR}}(\beta \Gamma)  := \ln \cosh\left(\beta \Gamma\right) $$
is expected to occur. The precise location of this first-order transition and the shape of the phase diagram of the QREM has been predicted by Goldschmidt~\cite{Gold90} in the 1990s on the bases of 
arguments using the replica trick and
the so-called static approximation in the associated path integral. His calculations have been repeated and refined in various papers - all still based on the replica trick and further approximations \cite{TD90,ONS07} (see also \cite{SIC12} and references). 
As a main result of this paper, we give a rigorous proof of this result. 

 \begin{theorem}\label{thm:main}
For any $ \Gamma , \beta \geq 0 $ almost surely:
$$
\lim_{N\to \infty} p_N(\beta, \Gamma) = \max\{ p^{\mathrm{REM}}(\beta)  , p^{\mathrm{PAR}}(\beta \Gamma) \} . 
$$
\end{theorem}

As will become clear from the proof, which is found in Section~\ref{sec:proof} below, the special structure of the pressure as a maximum of competing extremal cases is mainly caused by the fact that the REM's energy lanscape
is steep and rough due to the lack of correlations. This renders the model solvable. 
Before diving into the details of the proof, let us add some comments (see also Figure~\ref{fig:phase}): 
\begin{enumerate}
\item As in the classical case, the pressure $ p_N(\beta, \Gamma) $ is self-averaging, i.e.\  in the thermodynamic limit it coincides with its probabilistic average, the so-called quenched pressure $ \mathbb{E}\left[p_N(\beta, \Gamma) \right] $. For the QREM, this follows immediately from the Gaussian concentration inequality for Lipschitz functions. The Lipschitz constant of the pressure's variations with respect to the i.i.d.\ standard Gaussian variables $ g(\sigma) $ is bounded by
$$\displaystyle \sum_{\sigma \in \mathcal{Q}_N} \left(\frac{\partial p_N(\beta, \Gamma)}{\partial g(\sigma)} \right)^2 = \frac{\beta^2}{N \, 2^{2N} Z(\beta, \Gamma)^2 } \sum_\sigma \langle \sigma | e^{-\beta H} | \sigma \rangle^2 \leq \frac{\beta^2}{N}. $$
Here and in the following we use backet notation for matrix elements. Consequently, we have the Gaussian tail estimate 
\begin{equation}  \mathbb{P}\left(\left|  p_N(\beta, \Gamma)  - \mathbb{E}\left[p_N(\beta, \Gamma) \right]  \right|  > \frac{t \, \beta }{\sqrt{N} } \right) \leq  C \,  \exp\left(- c t^2\right)
\end{equation}
for all $ t > 0 $ and all $ N \in \mathbb{N} $ with some constants $ c, C \in (0,\infty) $. In fact, self-averaging for more general quantum $ p $-spin models has been stablished already in \cite{Craw07}.

\item For fixed $ \beta $ a first-order phase transition is found at 
$$ \Gamma_c(\beta) := \beta^{-1} \arcosh\left( \exp\left( p^{\mathrm{REM}}(\beta)\right)  \right) . $$ 
In particular, $ \Gamma_c(0) = 1  $ and $ \Gamma_c(\beta_c) = \beta_c^{-1} \arcosh(2) $. In the low-temperature limit,
$ \lim_{\beta \to \infty}  \Gamma_c(\beta) = \beta_c  $, 
the first-order transition connects to the known location of the quantum phase transition of the ground state~\cite{JKrKuMa08}.  
In this context, it is useful to recall  that 
the REM's extreme energies are almost surely found at $\| U \|_\infty = \beta_c N + o(N)$, cf.~\cite[Ch.~9]{Bov06}.
For $ \Gamma < \beta_c $, the energetically separated ground state is sharply localized near the lowest-energy configuration of the REM. 
For $ \Gamma > \beta_c $, the energetically separated ground state 
resembles the maximally delocalized state given by the ground state of $ T $. Near $ \Gamma = \beta_c $, the ground-state gap closes exponentially \cite{AW16}.
\item For $ \Gamma > \Gamma_c(\beta) $, the magnetization in the $ x $-direction is strictly positive,
$$
 \beta^{-1} \, \frac{\partial}{\partial  \Gamma} \, p^{\mathrm{PAR}}(\beta \Gamma) = \tanh(\beta J) > 0 . 
$$
\item  For all $ \Gamma < \Gamma_c(\beta) $ the line of the freezing transition transition remains unchanged at $ \beta = \beta_c $.  In the frozen regime, the QREM has zero entropy. 

\end{enumerate}

\subsection{Comments and open problems}

We close the introduction with some further comments and open problems:

\begin{enumerate}
\item For the quantum $ p $-spin model it is conjectured that  the structure of the phase diagram in Figure~\ref{fig:phase}  only changes smoothly in $ 1/p$  at low temperatures (see e.g.\  \cite{TD90} ).
Non-rigorous $ 1/p $ expansions in a replica analysis have been the basis of these assertions.
%
%

Such expansion-based arguments have been extended in~\cite{ONS07} to cover the case of ferromagnetic bias, in which the Gaussian spin-$p$ couplings are tilted towards a ferromagnetic interaction. The paper \cite{ONS07} argues that the spin glass 
phase will also disappear in favour of a ferromagnetic phase for sufficiently large tilting.

\item As in the classical case, the quenched pressure $ \mathbb{E}\left[p_N(\beta, \Gamma) \right] $ is generally smaller than the annealed pressure $ N^{-1} \ln \mathbb{E}\left[Z(\beta,\Gamma)\right] $. However, in the high-temperature phase, $ \beta < \beta_c $, asymptotic equality holds -- even in the quantum case as is not hard to show by performing the annealed average in the path-integral representation.
The fluctuation properties of the partition function are well studied in classical cases (see e.g.~\cite{ALR87,BKL02} and~\cite[Ch.~9-10]{Bov06} for further references). 
We leave it to some future work to extend these results to the quantum case. 

\item For a large class of mean-field spin glasses, the pressure in the thermodynamic limit is known to be universal in that it does not depend on the details of the randomness (cf.~\cite{Tal11} and references therein).  Such universality results have been extended to the quantum case in \cite{Craw07}.

\end{enumerate}

\section{Proof}\label{sec:proof}

The proof of Theorem~\ref{thm:main} consists of a pair of asymptotically coinciding upper and lower bounds.
\begin{proof}[Proof of Theorem~\ref{thm:main}]
The assertion is a consequence of Lemma~\ref{lem:lb} and Corollary~\ref{cor:up} below.
\end{proof}
\noindent
The following two subsections contain  the details of the argument.

\subsection{Lower bound}
Not surprisingly, our lower bound is more robust and will hold for more general $ p $-spin models also. 
Let us first recall that if $ U(\sigma) $ is a Gaussian random field of the from~\eqref{eq:spinp} with $ p \in [1,\infty] $, then its pressure 
\begin{equation}\label{eq:Parisi}
 p^{\mathrm{U}}(\beta)  := \lim_{N\to \infty } N^{-1} \ln 2^{-N} \sum_{\sigma \in \mathcal{Q}_N} e^{-\beta U(\sigma) } 
 \end{equation}
is known to converge almost surely to a non-random expression, which is in fact given by the famous Parisi formula \cite{Par80,Tal06,Tal11,Pan14}. In the special case $ p= \infty $ this reduces 
to $  p^{\mathrm{U}}(\beta) = p^{\mathrm{REM}}(\beta) $. 
\begin{lemma}\label{lem:lb}
Consider the quantum $ p $-spin model, i.e.\ $ H = \Gamma \, T + U $  with $ U $ diagonal and Gaussian of the from~\eqref{eq:spinp} with $ p \in [1,\infty] $. For any $ \Gamma , \beta \geq 0 $ and almost surely
\begin{equation}
\liminf_{N\to \infty} p_N(\beta, \Gamma) \geq  \max\{ p^{\mathrm{U}}(\beta)  , p^{\mathrm{PAR}}(\beta \Gamma) \} .
\end{equation}
\end{lemma}
\begin{proof}
We use the Gibbs variational
principle, 
\begin{equation}\label{eq:Gibbs}
 \ln \tr e^{-\beta H} = -  \inf_{\varrho} \left[  \beta \, \tr \left(H \varrho \right) +  \tr\left( \varrho \ln \varrho \right) \right] 
\end{equation}
in which the infimum is taken over all density matrices, $ \varrho \geq 0 $, $ \tr \varrho = 1 $, on $ \ell^2(\mathcal{Q}_N) $. There are two natural choices:
\begin{enumerate}
\item We may pick $ \varrho = e^{-\beta U }/ \tr e^{-\beta U } $. 
In this case, the right-hand side  is lower bounded by
$$
  \ln \tr e^{-\beta U } - \beta \Gamma  \tr \left(T \,  \varrho \right) =   \ln \tr e^{-\beta U } .
$$
The last step follows from the fact that the diagonal matrix elements of $ T $ vanish. Consequently, we arrive at the  bound,
$$
p_N(\beta, \Gamma)  \geq \frac{1}{N} \ln\left( \frac{1}{2^N} \sum_{\sigma \in \mathcal{Q}_N} e^{-\beta U(\sigma) }\right) ,
$$
which together with the known convergence~\eqref{eq:Parisi} yields the first part of the claim.
\item 
We may also pick $ \varrho = e^{-\beta \Gamma T  }/ \tr e^{-\beta \Gamma T } $. In this case, the right-hand side in \eqref{eq:Gibbs} reduces to
$$
\ln  \tr e^{-\beta \Gamma T }  - \beta \,   \tr \left(U \varrho \right) = N  \ln\left( 2\cosh(\beta \Gamma) \right) - \frac{\beta }{2^N} \sum_{\sigma \in \mathcal{Q}_N} U(\sigma) , 
$$
where we used  $ \langle \sigma | e^{-\beta \Gamma T } | \sigma \rangle = \cosh(\beta \Gamma)^N $ for the diagonal matrix element of the semigroup generated by $- T$. Consequently, we arrive at the  bound,
$$
p_N(\beta, \Gamma)  \geq p^{\mathrm{PAR}}(\beta \Gamma)  -  \frac{\beta }{N 2^N} \sum_{\sigma \in \mathcal{Q}_N} U(\sigma)  . 
$$
The last term converges to zero almost surely by the strong law of large numbers. More precisely, for any $ \varepsilon > 0 $, an exponential Chebychev bound yields
\begin{align*}
\mathbb{P}\left( \frac{1 }{N 2^N} \sum_{\sigma \in \mathcal{Q}_N} U(\sigma) > \varepsilon \right) & \leq  e^{- N \varepsilon^2/2 } \, \mathbb{E}\left[ \exp\left( \frac{ \varepsilon}{2^{N+1}} \sum_{\sigma \in \mathcal{Q}_N} U(\sigma) \right) \right] \\
& =  e^{- N \varepsilon^2/2 }  \exp\left(  \frac{\varepsilon^2}{2^{2(N+1) }} \sum_{\sigma, \sigma'} N \,  \xi_p(\sigma,\sigma') \right)  \leq  e^{- N \varepsilon^2/4 } . 
\end{align*}
The same bound also applies to $ - \sum_{\sigma}  U(\sigma) $. 
Since the right-hand side is summable in $ N $, a Borel-Cantelli argument ensures the claimed almost-sure convergence.
\end{enumerate}
\end{proof}
\subsection{Upper bound}
Typical values of the REM $ U(\sigma) $ fluctuate on order $  \mathcal{O}(\sqrt{N}) $. 
Our upper bound rests on the observation that configurations on which large negative deviations occur,
\begin{equation}\label{eq:LD}
\mathcal{L}_\varepsilon := \left\{ \sigma \in  \mathcal{Q}_N  \, \big| \, U(\sigma) \leq - \varepsilon N \right\} ,
\end{equation}
do not percolate even for  $ \varepsilon > 0 $ arbitrarily small. More precisely, the size of the maximally edge-connected components remains bounded uniformly in $ N $. 
For the precise formulation of this result, it is useful to recall that the Hamming distance
$$
 d(\sigma, \sigma') := \sum_{j=1}^N 1\left[\sigma_j \neq \sigma_j' \right] 
$$
renders $ \mathcal{Q}_N $ (through the nearest-neighbor relation)  into a graph called the Hamming cube, in which each vertex has exactly $ N $ neighbors.  
For future purposes, we also introduce the Hamming ball of  radius $ r \in [0, N] $ centered at $ \sigma \in  \mathcal{Q}_N $,
$$
B_r(\sigma) := \left\{ \sigma' \in  \mathcal{Q}_N \, \big| \,  d(\sigma, \sigma') \leq r \right\} .
$$
Its volume $ |B_{r} | $ is known to be bounded by $ \exp\left( N \gamma(r/N) \right) $ for all $ r < N/2 $ in terms of the binary entropy, $ \gamma(\xi) := -\xi \ln \xi - (1-\xi) \ln(1-\xi) $. 
Here, a simpler bound is sufficient:
\begin{equation}\label{eq:boundball}
\left| B_r \right| = \sum_{j=0}^r \binom{N}{j} \leq  \sum_{j=0}^r \frac{N^j}{j!} \leq e \, N^r .
\end{equation}
. 

\begin{definition}
An edge-connected component  $   \mathcal{C}_\varepsilon \subset \mathcal{L}_\varepsilon $ is a subset
for which each pair $ \sigma, \sigma' \in  \mathcal{C}_\varepsilon $ is connected through a connected edge-path of adjacent edges. 
An edge-connected component $  \mathcal{C}_\varepsilon $ is maximal if there is no other vertex $ \sigma \in  \mathcal{L}_\varepsilon \backslash  \mathcal{C}_\varepsilon $ such that $  \mathcal{C}_\varepsilon \cup \{ \sigma \} $ 
forms an edge-connected component. 
\end{definition}

For each realisation of the randomness the large-deviation set then naturally decomposes into a finite (edge-)disjoint union of maximally edge-connected components,
$$
\mathcal{L}_\varepsilon = \bigcup_\alpha  \, \mathcal{C}_\varepsilon^{(\alpha)} . 
$$
On any edge-connected  component $  \mathcal{C}_\varepsilon  $ for every vertex $ \sigma \in \mathcal{C}_\varepsilon $ there is some $  \sigma' \in  \mathcal{C}_\varepsilon \backslash \{\sigma\}   $ with $  d(\sigma, \sigma')  \in \{ 1, 2 \}$ -- not necessarily $  d(\sigma, \sigma')  = 1 $. By construction, we thus have for all $ \alpha \neq \alpha' $:
\begin{equation}\label{eq:disjoint}
 d\left( \mathcal{C}_\varepsilon^{(\alpha)},  \mathcal{C}_\varepsilon^{(\alpha')} \right) = \min\left\{ d(\sigma, \sigma') \, | \, \sigma \in \mathcal{C}_\varepsilon^{(\alpha)} \wedge \sigma' \in \mathcal{C}_\varepsilon^{(\alpha')} \right\}  > 2 . 
 \end{equation}
The next lemma controls with good probability the size of each subset $  \mathcal{C}_\varepsilon^{(\alpha)}$, which is just the number of its vertices and denoted by $ |  \mathcal{C}_\varepsilon^{(\alpha)} | $. 
\begin{lemma}\label{lem:goodset}
For all $ \varepsilon > 0 $ and $ N \in \mathbb{N} $ there is some subset $ \Omega_{\varepsilon, N} $ of realizations such that:
\begin{enumerate}
\item for some $ c_\varepsilon > 0 $, which is independent of $ N $, and all $ N $ large enough:
$$
\mathbb{P}\left(  \Omega_{\varepsilon, N} \right) \geq 1 - e^{- c_\varepsilon N } ,
$$
\item on $ \Omega_{\varepsilon, N} $:
$ \displaystyle \, 
\max_\alpha \big|  \mathcal{C}_\varepsilon^{(\alpha)} \big| < K_\varepsilon := \left\lceil \frac{4 \ln 2}{\varepsilon^2}   \right \rceil $.
\end{enumerate}
\end{lemma}
\begin{proof}
We start from noting that the event
\begin{equation}
\Omega_{\varepsilon, N} := \bigcap_{\sigma_\in   \mathcal{Q}_N} \left\{ \left| B_{r_\varepsilon }(\sigma) \cap  \mathcal{L}_\varepsilon  \right| < K_\varepsilon \right\} 
\end{equation}
with $ r_\varepsilon := 4K_\varepsilon $ 
implies the second assertion in the lemma. This follows from the fact that in the event $ \Omega_{\varepsilon, N} $, in which there are at most $ K_\varepsilon - 1 $ large deviation sites in the ball of 
radius $ r_\varepsilon $ around any fixed $ \sigma \in  \mathcal{L}_\varepsilon $, the edge-connected component 
to which $ \sigma $ belongs, must be strictly contained in a ball of radius at most $ 2 (K_\varepsilon - 1 ) < r_\varepsilon -2 $, i.e. it cannot edge-connect to other vertices outside the ball $B_{r_\varepsilon }(\sigma) $ and hence consists of at most $ K_\varepsilon $ vertices.

It therefore remains to estimate the probability of the event complementary to $ \Omega_{\varepsilon, N} $.
Using the union bound we obtain: 
\begin{align}
\mathbb{P}\left( \bigcup_{\sigma_\in  \mathcal{Q}_N}  \left\{ \left| B_{ r_\varepsilon}(\sigma) \cap  \mathcal{L}_\varepsilon  \right| \geq K_\varepsilon \right\}  \right) 
& \leq \! \sum_{\sigma \in \mathcal{Q}_N}  \mathbb{P}\left(  \left| B_{ r_\varepsilon}(\sigma) \cap  \mathcal{L}_\varepsilon  \right| \geq K_\varepsilon  \right) \notag \\
& \leq  \! \sum_{\sigma \in \mathcal{Q}_N}  \sum_{j= K_\varepsilon}^{|B_{ r_\varepsilon} |}  \mathbb{P}\left(  \left| B_{ r_\varepsilon}(\sigma) \cap  \mathcal{L}_\varepsilon  \right| = j  \right) \notag \\
& \leq 2^N  \sum_{j= K_\varepsilon}^{|B_{ r_\varepsilon} |}  \binom{|B_{ r_\varepsilon} |}{j} e^{- j \varepsilon^2 N /2} \leq 2^N \sum_{k=K_\varepsilon}^\infty \frac{|B_{ r_\varepsilon} |^j}{j!} e^{-j \varepsilon^2 N /2} \notag \\
& \leq 2^N \frac{|B_{ r_\varepsilon} |^{K_\varepsilon}}{K_\varepsilon!} e^{-K_\varepsilon \varepsilon^2 N /2}  \exp\left( |B_{ r_\varepsilon} | e^{- \varepsilon^2 N /2 }\right)  \notag \\
& \leq \frac{|B_{ r_\varepsilon} |^{K_\varepsilon}}{K_\varepsilon!}  e^{-K_\varepsilon \varepsilon^2 N /4} \exp\left( |B_{ r_\varepsilon} | e^{- \varepsilon^2 N /2 }\right)  . \label{eq:long}
\end{align}
Here the third line relies on the fact that the number of subsets of a given size equals the binomial coefficient. Moreover, specifying the large-deviation sites in $ B_{ r_\varepsilon}(\sigma) $ allows one to compute the probability of the event using the independence of the random field $ U(\sigma) $. To estimate this probability, we use the elementary estimate on the complementary error function, 
\begin{equation}
\mathbb{P}\left( \sigma \in   \mathcal{L}_\varepsilon \right) = \int_{-\infty}^{-\varepsilon \sqrt{N} } e^{-x^2/2} \frac{dx}{\sqrt{2\pi}} \leq e^{-\varepsilon^2 N /2} , 
\end{equation}
as well as the trivial bound on the probability of the complementary elementary event.  The last inequality in the second line of~\eqref{eq:long} results from a simple bound on the binomial coefficient. 
The forth  line is the standard estimate of the remainder of the exponential series. Finally, the last line follows by definition of $ K_\varepsilon $. 
Since the volume of the ball $ |B_{ r_\varepsilon} | $ grows only polynomially in $ N $ by~\eqref{eq:boundball}, the right-hand side of~\eqref{eq:long} is exponentially bounded for large enough $ N $. 
This completes the proof.
\end{proof}

Our main idea behind an upper bound on the partition function $  Z(\beta, \Gamma)  $ is  to decompose $ H $ into the multiplication operator $ U $ restricted to vertices 
in $ \mathcal{L}_\varepsilon $ and the QREM $ H $ restricted to the complementary set $ \mathcal{L}_\varepsilon^c $ plus a remainder term $ A_{ \mathcal{L}_\varepsilon} $.
For this purpose, we write $ \ell^2(\mathcal{Q}_N) =   \ell^2( \mathcal{L}_\varepsilon) \oplus  \ell^2( \mathcal{L}_\varepsilon^c)  $ and set $ U_{\mathcal{L}_\varepsilon} $ the 
multiplication operator by the REM values on $  \ell^2( \mathcal{L}_\varepsilon) $. On the orthogonal complement $  \ell^2( \mathcal{L}_\varepsilon^c)   $, we define
the natural restriction of \eqref{eq:Ham}. Note that $-T $ is the adjacency matrix on the Hamming cube. 
In the restriction $ H_{\mathcal{L}_\varepsilon^c} $, we simply restrict the adjacency matrix to the subgraph associated with $ \mathcal{L}_\varepsilon^c $. 
We then define $ A_{ \mathcal{L}_\varepsilon} $ through:
\begin{equation}\label{eq:decomp}
 H =: U_{\mathcal{L}_\varepsilon} \oplus H_{\mathcal{L}_\varepsilon^c} - \Gamma  A_{ \mathcal{L}_\varepsilon} . 
\end{equation}
Clearly, the matrix elements of the remainder term are related to all edges reaching $  \mathcal{L}_\varepsilon $:
\begin{equation}\label{eq:matrixel}
\langle \sigma |  A_{ \mathcal{L}_\varepsilon} | \sigma' \rangle = 
\begin{cases} 1 & \mbox{if $ \sigma \in  \mathcal{L}_\varepsilon $ or $ \sigma' \in   \mathcal{L}_\varepsilon $ and $ d(\sigma, \sigma') = 1 $,} \\
			0 & \mbox{else.}
			\end{cases} 	
\end{equation}

The following lemma contains an estimate on the operator norm of the remainder. In case the components in the decompositions are of small size, this estimate is not so wasteful.
\begin{lemma}\label{lem:opbd}
Let $ \mathcal{L}_\varepsilon = \bigcup_\alpha  \, \mathcal{C}_\varepsilon^{(\alpha)} $ stand for a finite (edge-)disjoint union of maximally edge-connected components of the large deviation set~\eqref{eq:LD}
Then 
\begin{equation}
\left\| A_{ \mathcal{L}_\varepsilon} \right\| = \sqrt{2N \,  \max_{\alpha} \big|  \mathcal{C}_\varepsilon^{(\alpha)} \big| } . 
\end{equation}
\end{lemma}
\begin{proof} 
Since the components are edge-disjoint in the sense that \eqref{eq:disjoint} holds, we have 
$$ 
\left\| A_{ \mathcal{L}_\varepsilon} \right\| = \max_{\alpha}  \big\|  A_{ \mathcal{C}_\varepsilon^{(\alpha) }}  \big\| ,
$$
where the operators in the right-hand side satisfy~\eqref{eq:matrixel} with $    \mathcal{L}_\varepsilon $ substituted by $ \mathcal{C}_\varepsilon^{(\alpha) }$. 
Consequently, their operator norms are bounded by a Frobenius estimate
$$
 \big\|  A_{ \mathcal{C}_\varepsilon^{(\alpha) }}  \big\|  \leq \sqrt{\sum_{\sigma, \sigma' }  \left| \langle \sigma |  A_{\mathcal{C}_\varepsilon^{(\alpha)}}  | \sigma' \rangle \right|^2} . 
$$ 
Since the double sum is restricted to $ \sigma \in  \mathcal{C}_\varepsilon^{(\alpha) } $ or $ \sigma' \in  \mathcal{C}_\varepsilon^{(\alpha) } $ and, in each of the two cases, the other sum has at most $ N $ terms,
the assertion follows. 
\end{proof}
The fact that the operator norm in the preceding lemma does not scale with $ N $ might sound remarkable at first sight. However, we remind the reader that 
even the full adjacency matrix $ -T_{B_{N\rho}} $ restricted to a Hamming ball of radius $N\rho $ with $ \rho \in (0,1/2) $, is known \cite{FriedTill05} to be bounded by
$\big\| T_{B_{N\rho}} \big\| \leq 2 N\sqrt{\rho (1-\rho)} +o(N) $.

We are now ready to conclude our asymptotically sharp upper bound. 
\begin{corollary}\label{cor:up} 
For any $ \Gamma , \beta \geq 0 $ almost surely:
$$
\limsup_{N\to \infty}  p_N(\beta,\Gamma) \leq \max\left\{ p^{\mathrm{REM}}(\beta),   p^{\mathrm{PAR}}(\beta \Gamma)   \right\} \, . 
$$
\end{corollary} 
\begin{proof}
We pick $ \varepsilon > 0 $ arbitrarily small and start from the decomposition~\eqref{eq:decomp} of the Hamiltonian. The Golden-Thompson inequality yields
\begin{align}
Z(\beta, \Gamma) & \leq  2^{-N}\,  \tr e^{-\beta U_{\mathcal{L}_\varepsilon} \oplus H_{\mathcal{L}_\varepsilon^c} } \, e^{ -\beta A_{ \mathcal{L}_\varepsilon}}  \notag \\
	& \leq2^{-N}\,  e^{ \beta \| A_{ \mathcal{L}_\varepsilon} \|  }  \left(  \tr_{\ell^2(\mathcal{L}_\varepsilon)} e^{-\beta U_{\mathcal{L}_\varepsilon} } +  \tr_{\ell^2(\mathcal{L}_\varepsilon^c)}  e^{-\beta H_{\mathcal{L}_\varepsilon^c}} \right)  . \notag 
\end{align}
The first term in the bracket on the right-hand side is trivially estimated in terms of the partition function of the REM:
$$
2^{-N}\,   \tr_{\ell^2(\mathcal{L}_\varepsilon)} e^{-\beta U_{\mathcal{L}_\varepsilon} } \leq Z(\beta,0) = e^{N p_N(\beta,0) } . 
$$
For the second term we use the fact that the adjacency matrix $ - T_{\mathcal{L}_\varepsilon^c} $ has non-negative matrix elements and hence generates a positivity preserving semigroup on $ \ell^2(\mathcal{L}_\varepsilon^c)$.  Since the diagonal values 
of its perturbation are bounded from below by $ - \varepsilon N $ by assumption on $ \mathcal{L}_\varepsilon^c $, we conclude
\begin{align}
2^{-N}\, \tr_{\ell^2(\mathcal{L}_\varepsilon^c)} e^{-\beta H_{\mathcal{L}_\varepsilon^c}} & \leq e^{\beta \varepsilon N} 2^{-N}\, \tr_{\ell^2(\mathcal{L}_\varepsilon^c)} e^{-\beta J T_{\mathcal{L}_\varepsilon^c} } \notag \\
& \leq e^{\beta \varepsilon N} 2^{-N}\, \tr e^{-\beta J T} =  \exp\left( N \left( \beta \varepsilon +  p^{\mathrm{PAR}}(\beta \Gamma)  \right)  \right) . \notag 
\end{align}
Here the last inequality is the monotonicity of the trace in the domain, which is a consequence of the non-negativity of the matrix elements of the adjacency matrix. 
To summarize, we thus obtain
\begin{equation}
 p_N(\beta,\Gamma) \leq \max\left\{ p_N(\beta,0) , \beta \varepsilon +  p^{\mathrm{PAR}}(\beta \Gamma)   \right\} + \tfrac{ 1}{N} \left( \beta \,  \| A_{ \mathcal{L}_\varepsilon} \| + \ln 2 \right) .
\end{equation}
According to Lemma~\ref{lem:goodset} there is some $ \Omega_{\varepsilon , N} $ whose complementary probability is exponentially small in $ N $ and on which Lemma~\ref{lem:opbd} guarantees that  for all $ N $ large enough:
$$ p_N(\beta,\Gamma) \leq \max\left\{ p_N(\beta,0) , p^{\mathrm{PAR}}(\beta \Gamma)   \right\}  + 2  \beta \varepsilon \, . $$
Since the probabilities of the complementary event are summable in $ N $, a Borel-Cantelli argument together with the known almost sure convergence~\eqref{eq:REMc} of the REM thus finishes the proof.
\end{proof}

\minisec{Acknowledgments}
This work was supported by the DFG under EXC-2111 -- 390814868. 


\bigskip
\bigskip
\begin{minipage}{0.5\linewidth}
\noindent Chokri Manai and Simone Warzel\\
MCQST \& Zentrum Mathematik \\
Technische Universit\"{a}t M\"{u}nchen\\
Corresponding author: \verb+warzel@ma.tum.de+
\end{minipage}
\end{document}